\documentclass{article}
\usepackage{ijcai22}

\usepackage{times}
\usepackage{soul}
\usepackage{url}
\usepackage[hidelinks]{hyperref}
\usepackage[utf8]{inputenc}
\usepackage[small]{caption}
\usepackage{graphicx}
\usepackage{amsmath}
\usepackage{amsthm}
\usepackage{booktabs}
\usepackage{algorithm}
\usepackage{algorithmic}
\newif\ifijcai
\title{Anytime Capacity Expansion in Medical Residency Match by Monte Carlo Tree Search\thanks{Atsushi Iwasaki was supported by JSPS KAKENHI Grant Numbers JP21H04890 and JP20K20752.
\ifijcai
A full version of this paper is found in \url{https://arxiv.org/abs/2202.06570}.
\else
\fi
}}

\author{
Kenshi Abe$^1$\and
Junpei Komiyama$^2$\and
Atsushi Iwasaki$^3$
\affiliations
$^1$CyberAgent, Inc., $^2$New York University, $^3$University of Electro-Communications
\emails
abe\_kenshi@cyberagent.co.jp, junpei@komiyama.info, atsushi.iwasaki@uec.ac.jp
}

\usepackage{mathtools}
\usepackage{amsmath}
\usepackage{amsthm}
\usepackage{amssymb}
\mathtoolsset{showonlyrefs}  

\usepackage{xcolor}

\usepackage{amsthm}
\usepackage{amsfonts}
\usepackage{amssymb} %\checkmark
\usepackage{bbm}
\usepackage{graphics}
\usepackage{enumerate}
\usepackage{float}
\usepackage{caption}
\usepackage{subcaption}

\theoremstyle{remark}

\theoremstyle{definition}
\newtheorem{thm}{Theorem}

\newtheorem{definition}{Definition}
\newtheorem{prop}[thm]{Proposition}

\newtheorem{remark}{Remark}

\newtheorem{assp}{Assumption}

\usepackage{amsmath} 
\usepackage{amssymb}
\usepackage{bm}
\usepackage{ascmac}
\usepackage{setspace}
\usepackage[at]{easylist}

\usepackage{amsmath}
\DeclareMathOperator*{\argmax}{arg\,max}

\newcommand{\Real}{\mathbb{R}}
\newcommand{\Natural}{\mathbb{N}}

\newcommand{\Prob}{\mathbb{P}}

\newcommand{\Ind}{\bm{1}}

\newcommand{\ED}{\mathcal{D}}
\newcommand{\EE}{\mathcal{E}}

\newcommand{\EH}{\mathcal{H}}

\newcommand{\EM}{\mathcal{M}}

\newcommand{\ET}{\mathcal{T}}

\newcommand{\ETall}{\mathcal{T}_{\mathrm{all}}}
\newcommand{\UCB}{\mathrm{UCB}}

\newcommand{\bsucc}{\bm{\succ}}

\newcommand{\bt}{\bm{t}}
\newcommand{\bx}{\bm{x}}

\newcommand{\bp}{\bm{p}}
\newcommand{\bq}{\bm{q}}

\newcommand{\hatmu}{\hat{\mu}}

\newcommand{\rank}{\mathrm{rank}}

\newcommand{\Conf}{C_p}

\allowdisplaybreaks

\begin{document}

\maketitle

\begin{abstract}
This paper considers the capacity expansion problem in two-sided matchings, where the policymaker is allowed to allocate some extra seats as well as the standard seats. % 
% This paper deals with two-sided matchings where one side (doctors or students) is matched to the other side (hospitals or colleges). 
% 
In medical residency match, each hospital % is assumed to 
accepts a limited number of doctors. %The limit, i.e., 
Such capacity constraints are typically given in advance. However, such exogenous constraints can compromise the welfare of the doctors; some popular hospitals inevitably dismiss some of their favorite doctors. Meanwhile, it is often the case that the hospitals are also benefited to accept a few extra doctors. 
% who are very enthusiastic. 
% In the light of these observations, 
% we consider the capacity expansion problem in two-sided matchings, where the policymaker is allowed to allocate some extra seats as well as the standard seats. 
% While such a problem is known to be NP-complete (Bobbio et al. 2021), exploiting the structure of the problem can lead to an efficient method that finds a reasonable solution in a practical computational time. 
To tackle the problem, we propose an anytime method that 
the upper confidence tree searches the space of capacity expansions, each of which has a resident-optimal stable assignment that the deferred acceptance method finds.
% combining the upper confidence tree (UCT) search with the deferred acceptance (DA) method. 
% UCT searches the space of capacity expansions, each of which has a resident-optimal stable assignment that DA finds.
Constructing a good search tree representation significantly boosts the performance of the proposed method. Our simulation shows that the proposed method identifies an almost optimal capacity expansion with a significantly smaller computational budget than exact methods based on mixed-integer programming.
\end{abstract}

\section{Introduction}
\label{sec:introduction}
This paper considers the capacity expansion problem in two-sided matchings,
where the policymaker is allowed to allocate some extra seats as well as the standard seats.
The two-sided matchings have a lot of real applications such as medical residency match \cite{roth:aer:1991} and school choice \cite{Abdulkadiroglu:aer:2003}.
The theory has been extensively developed across computer science and economics~\cite{Roth:CUP:1990,manlove:2013}.
% . See the book by
% \cite{Roth:CUP:1990} or \cite{manlove:2013} for a comprehensive survey.
To establish our model and concepts, we use medical residency match as a running example.
In this literature, each hospital (school) accepts a limited number of residents (students). 
Such capacity constraints (or quotas) are assumed to be known and fixed in advance.

However, in practice, even for hospitals, or its stakeholders, it is often uncertain 
how the capacity constraints are specified beforehand. Those could be flexible and variable
more than the previous studies preclude. 
By pooling some extra funding, hospitals can accept a few more residents. % than the predefined number. 
Schools may be granted some budgets or dispatched teachers from their district according to the demand.

Such a situation involves a lot of resource allocation problems where capacities are predefined,
e.g., how many jobs each machine takes in the jobshop scheduling problem.
How we should choose the capacities has been paid less attention than it deserves, though
they influence the performance generated by allocations. 
% This paper endogenizes these capacities. 
% In many cases, capacity
% expansion is computationally challenging (NP-Complete)
% combinatorial optimization, and using UCT gives a reasonable
% solution as we demonstrate in a two-sided matching
% problem.

A classical comparative statistic already analyzes the effect of such expansion. % In two-sided matchings, 
If the capacity of some hospitals is expanded, the welfare of every resident weakly improves~\cite{%kelso:ecma:1982,gale:dam:1985,
Roth:CUP:1990}. However, it is rarely investigated how capacities should be expanded to improve the total welfare. 
% Meanwhile, the hospitals are also benefited to accept a few extra doctors. 
% 
\citeauthor{bobbio2021capacity}~\shortcite{bobbio2021capacity} have initiated the question of how limited extra seats should be allocated, keeping resulting matchings \textit{stable}. The stability is a key concept for the two-sided matching and if a matching is stable,
no groups of residents and hospitals have profitable deviations. The celebrated deferred acceptance (DA) algorithm is known to find a stable matching~\cite{galeshapley}.  
The capacity expansion with DA, i.e., finding the optimal allocation of extra seats in the welfare of residents, 
is formulated as an integer quadratic programming~\cite{bobbio2021capacity}, which is computationally challenging to solve exactly.
%It minimizes the total rankings of the matched residents and is computationally hard.
In addition to the exact method, they also developed some greedy heuristics, which run very effectively yet are suboptimal.
%However, since the exact method solves mixed integer programming, it is time-consuming. 
%The heuristics do not guarantee the optimal solution and often falls into a local optimum. 

\if0
\begin{table}[t!]
\begin{center}
\caption{Comparison of methods for the capacity expansion.
}
\label{tbl_demand}
\begin{tabular}{lcccc} 
& Optimality & Anytime & Diversity \\ \hline
Greedy & %UCT
& 
\checkmark
\\
Exact & %Pruningと言っているタイプ
\checkmark
& 
\\
Our Method & %優先度最大の学校からバッチで割り当てる
\checkmark
& 
\checkmark
&
\checkmark
\\
\hline
\end{tabular}
\end{center}
\end{table}
\fi

%Given the hardness of the problem as well as somewhat ambiguous goal, 
\if0
The following properties are desired in solving the capacity expansion problem. 
\begin{enumerate}
    \item Optimality (asymptotic): It can spend the computational time flexibly. If a policymaker spends a sufficiently large computational budget, the method returns the optimal solution.
    \item Anytimeness: A policymaker can stop the computation anytime to receive a current best solution. This is important given the hardness of obtaining the optimal solution. 
    \item Diversity: The method visits and can recommend many (if suboptimal) solutions among which the policymaker can choose. This is particularly important when there are some implicit constraints/preferences.
\end{enumerate}
\fi
This paper proposes 
% Our contribution in this context is the proposal of 
an alternative method to solve the capacity expansion problem where the upper confidence tree (UCT) searches the space of capacity expansions.\footnote{An implementation of our method is available at \url{https://github.com/CyberAgentAILab/uct_capacity_expansion}.} Each pattern of them has a resident-optimal stable assignment that DA finds. Not only does UCT obtain an optimal solution given a sufficiently large time, but also a policymaker can stop UCT anytime to obtain a reasonably good solution, which is important given the hardness of obtaining an optimal solution. 
%Another advantage of UCT is that it visits many candidate solutions among which the policymaker can choose when they are aware of additional requirements. This is particularly important when there are some implicit constraints/preferences.
% 
% We show that how a tree search method performs crucially depends on the efficiency of the tree representation. We characterize a good representation that exploits the structure of the capacity expansion problem. 
We then characterize a good tree representation so that the tree search method efficiently exploits the structure of the capacity expansion problem.
% \subsection{Related literature on AI community}
%Note that % recently 
%JK: すこし文頭けずりました
There have been a certain amount of studies on two-sided matchings in the AI community, although this literature has been established mainly in fields across algorithms and economics%
% \citeauthor{drummond:ijcai:2013}~\shortcite{drummond:ijcai:2013,drummond:aaai:2014} examined preference elicitation procedures for two-sided matching. 
%In fact, % In the context of mechanism design, \citeauthor{hosseini:aaai:2015}~\shortcite{hosseini:aaai:2015} considered a mechanism for a situation where agents’ preferences dynamically change.
% 
% Alternatively 
%matching with constraints is prominent% across computer science and economics%
~\cite{kamada:aer:2015,biro:tcs:2010,fragiadakis::2012,Goto:aamas:2014,Goto:aij:2016}. %
In many application domains, various distributional constraints are imposed on an outcome, e.g., regional maximum quotas are imposed 
on hospitals in urban areas to allocate more doctors to
rural areas~\cite{kamada:aer:2015}.
% or minimum quotas are imposed when 
% school districts require that 
% at least a certain number of students are allocated to each school
% to enable these school to operate properly%
% ~\cite{biro:tcs:2010,fragiadakis::2012,Goto:aamas:2014}.
% % Furthermore, Goto et al.~\cite{Goto:aij:2016} handle regional minimum/maximum quotas simultaneously. 
% % 
% In addition, another type of distributional constraint involves  
% \textit{diversity constraints} in school choice programs~\cite{Abdulkadiroglu:aer:2003,kojima2012school,hafalir2013effective,ehlers::2012,kurata:jair:2017}. 
% They are implemented to give students/parents 
% an opportunity to choose which public school to attend. 
% However, a school is required to balance its composition of students,
% typically in terms of socioeconomic status.
% Controlled school choice programs must provide choices
% while maintaining distributional constraints. 

\section{Problem Setup}

%記号はほぼBobbio論文をフォロー予定 -> 制約を[0, bmax]にして拡張にする
% 一般化してこのモデルにする? https://www.jair.org/index.php/jair/article/view/11042/26214
% C（大学数）はBobbioだとUです
Let $\ED=\{d_1,d_2,\ldots,d_{D}\}$
% $\ED = \{i_1,i_2,\dots,i_S\}$ 
be the set of residents (doctors) and 
$\EH = \{h_1, h_2,\dots,h_{H}\}$ be the set of hospitals. 
We consider a bipartite graph where $\ED$ and $\EH$ are connected by edges $\EE$. Each edge represents that the resident is accepted by the hospital. 
% CA problemというのはBobbioの言い方なので変えたほうがいいかも
Residents and hospitals have their own preferences. We denote $h \succ_d h'$ if resident $d$ prefers hospital $h$ to hospital $h'$. Similarly, we denote $d \succ_h d'$ if hospital $h$ prefers resident $d$ over resident $d'$.
A vector $\bq = (q_1, q_2, \dots, q_H)$ represents the quota of the hospitals.
An instance of the residency match problem is a tuple $\Gamma = (\ED, \EH, \bsucc, \bq)$, where $\bsucc$ denotes all preferences of the residents and the hospitals.

% Matching
A matching $M \in \EM$ is a subset of $\EE$ such that the following criteria are met.
Each resident $d\in \ED$ has at most one edge. Each hospital $h\in \EH$ can have at most $q_h$ edges. We say that resident $d$ is unassigned if there is no edge from $d$.
%, where $q_h$ is an associated quota. 
We denote $M(d)$ and $M(h)$ to represent the hospital assigned to resident $d$ and the set of the residents assigned to hospital $h$, respectively.
% Bobbioはpairとedgeを両方使っている
Given a matching $M$, we say a pair $(d,h)\in \EE$ is a \textit{blocking pair} if (1) resident $d$ is unassigned or prefers hospital $h$ to $M(d)$, and (2) hospital $h$ is such that $|M(h)|<q_h$ or prefers resident $d$ over at least one resident in $M(h)$. Matching $M$ is stable if there is no blocking pair. 

\citeauthor{galeshapley}~\shortcite{galeshapley} showed that a stable matching always exists and proposed the deferred-acceptance (DA) algorithm that yields a stable matching. 
%（ここにweakly preferredなどのDAの性質について書きます）
The resident-proposing DA finds the resident-optimal stable matching where no resident is better off among all stable matchings~\cite{Roth:CUP:1990}. That matching is also obtained by 
minimizing the total resident ranking under the stability constraints.
A resident ranking is denoted by $\mathrm{rank}_{d}(h)$ the rank of hospital $h$ according to the resident $d$'s preference. This is also referred to as the \textit{cost} of the match $(d,h)$. 
Given a matching $M$, the total resident ranking is defined as 
\begin{equation}\label{eq_utility}
\sum_{(d,h) \in M} \rank_d(h). 
\end{equation}
% rankではなくて，doctors' utilitiesで書いた方がいいかなー
%In other words, a stable matching is doctor-optimal if and only if it is a stable matching where each doctor is assigned to the most preferred hospital among the possible stable matchings. 
%
%\begin{remark}
% The DA algorithm 
%
DA gives a matching that minimizes Eq.~\eqref{eq_utility} among all stable matchings \cite{bobbio2021capacity}.
%\end{remark}

% Following \cite{bobbio2021capacity}, we focus on stable matching where the doctor-side utility is maximized. As the doctor has ordinal preference, 
% \begin{equation}\label{eq_utility}
% \sum_{(i,j) \in M} \rank_i(j),
% \end{equation}
% where $\rank_i(j)$ is the position of hospital $h$ in the doctor $i$'s preference.

% %ここにDoctorから見たrankの最小化はGSで計算できるとか書きます：本当？・・・
% \begin{prop}{\rm (DA maximizes doctor utility)}\label{prop_dautil}
% %ここにdoctorから見たrankの最小化はGSで計算できるとか書きます
% \end{prop}

% \subsection{Capacity expansion}

% 目的関数（変えられると良いかも）
% 基本的に生徒側の希望を最大化しているので、生徒から見たutility maximizationになる？

This paper improves % considers the improvement of 
the resident utilities by allowing additional seats. Letting $\bt \in \Natural^\EH$ be a non-negative vector, we consider an expanded matching problem in which the capacity of each hospital $h$ is the sum of regular capacity $q_h$ and the extended capacity $t_h$.  
A reasonable extension of the capacity should not be very large and can accommodate the demand of each hospital. Given this, 
we consider the following optimization problem: 
%元論文のc_jはq_h?
% \right\}がちゃんと出ないけど後で直す．
Let 
\begin{align} %feasible set
    \mathcal{P}&=\Biggl\{ 
     (\bx,\bt)\in \{0,1\}^\mathcal{E}\times\Theta  \ \biggl\vert\ 
    \sum_{h\in \EH} x_{dh}\leq 1\ \forall d\in \ED,\ \\
    &\ \ \ \ \ \ \ \ \sum_{d\in \ED} x_{dh}\leq (q_{h}+t_{h})\ \forall h\in \EH\Biggr\}, \mathrm{and}\\
\Theta &= \left\{
\bt \in \{0,1,2,\dots,B\}^{\EH} 
\ \biggl\vert\ 
t_h \le b_h\ \forall h, \sum_{h \in \EH} t_{h} \le B
\right\}
\end{align}
be the set of matchings. 
The capacity expansion problem is the following: 
% AI: V_{dh}, T_{dh}は後で使わないなら定義を直接書いてもよい．
% optimization problem of
\begin{align}\label{opt_main}
    \min_{\bx,\bt} & \sum_{(d,h)\in\mathcal{E}} \mathrm{rank}_{d}(h)x_{dh}\\
    \mathrm{s.t.}\ & (q_h+t_h)\left(1-\sum_{h'\in S_{dh}}x_{dh'}\right)\leq \sum_{d'\in T_{dh}}x_{d'h}\ \forall (d,h)\in \EE,
    \\
    & (\bx,\bt)\in \mathcal{P}.
\end{align}
where $S_{dh} = \{h' \in \EH: \rank_d(h') \le \rank_d(h)\}$ is the set of indices of hopsitals that resident $d$ prefers at
least as hospital $h$; similarly, $T_{dh} = \{d' \in \ED: \rank_h(d') < \rank_h(d)\}$ is the set of indices of residents
that hospital $h$ prefers more than resident $d$. 
%JK: 書いたけどd,hはindexじゃないので\rankの使い方がちょっとabuseですね 
Namely, keeping stability, we maximize the residents' welfare subject to the capacity constraints of the hospitals, where the capacity can be relaxed up to $B$ seats. Moreover, we assume that each hospital $h$ has its own expansion limit~$b_h$.
%In summary, an instance of the capacity expansion problem is a tuple $\Gamma_{\mathrm{ex}} = (\ED, \EH, \bsucc, \bq, \bb, B)$, where $\bb = (b_1, b_2,\dots,b_n)$. Our goal is to derive an optimal solution of $\Gamma_{\mathrm{ex}}$.

\begin{remark}{\rm (Complete preference)}
The optimization of Eq.~\eqref{opt_main} requires a complete preference of all residents and hospitals, which is often impractical. For example, in the real data of a residency matching program, an average resident only applies to four or five hospitals. In this case, we can add a dummy hospital of infinite capacity and rank all unpreferred hospitals after the dummy hospital.
\end{remark}

\citeauthor{bobbio2021capacity}~\shortcite{bobbio2021capacity} proved the computational hardness of the problem of Eq.~\eqref{opt_main} without hospital-wise limit $b_h$.

\begin{prop}
{\rm (Hardness result~\cite{bobbio2021capacity})}
% {\rm (Hardness of capacity expanded matching \cite{bobbio2021capacity})}
\label{prop_hardness}
%Theorem 1 therein
The decision version of  Eq.~\eqref{opt_main} is NP-complete even in the case where there is no hospital-wise limit $b_h$.
\end{prop}

The key idea of our algorithm for this problem is that, when we fix an expansion vector $\bt \in \Theta$, the capacity expansion problem of Eq.~\eqref{opt_main} boils down to the standard matching instance $\Gamma = (\ED, \EH, \bsucc, \bq+\bt)$ that we can solve efficiently by DA.\footnote{The computational complexity of DA is $O((D+H)^2)$.} Therefore, a tree search on the space of $\Theta$ combined with DA is expected to give an efficient solution. 
%The performance of UCT depends on the structure of the tree, and next section 
%show that the ordering based on the popularity and the envy can be used to yield a mapping from $\Theta$ to a search tree.

%（これは↑のsetupに織り込めそうなのでたぶん消します）
%\textbf{Extension vector}
%Let $\bt \in [0, B]^\EH$ be an expansion. For a fixed expansion $\vecc$, the rank minimization problem is solved by the well-known deferred acceptance (DA) algorithm (todo cite). 
%However, the problem is significantly challenging given the variation of $\vecc$. 

% All-Path-Backpropagation (AMAF) についての説明
% これも図をかきます これは文章を書いてもわかりにくいので図+alg
\begin{figure*}
%\vspace{6em}
	\centering
	\includegraphics[bb=0 0 923 386, scale=0.40]{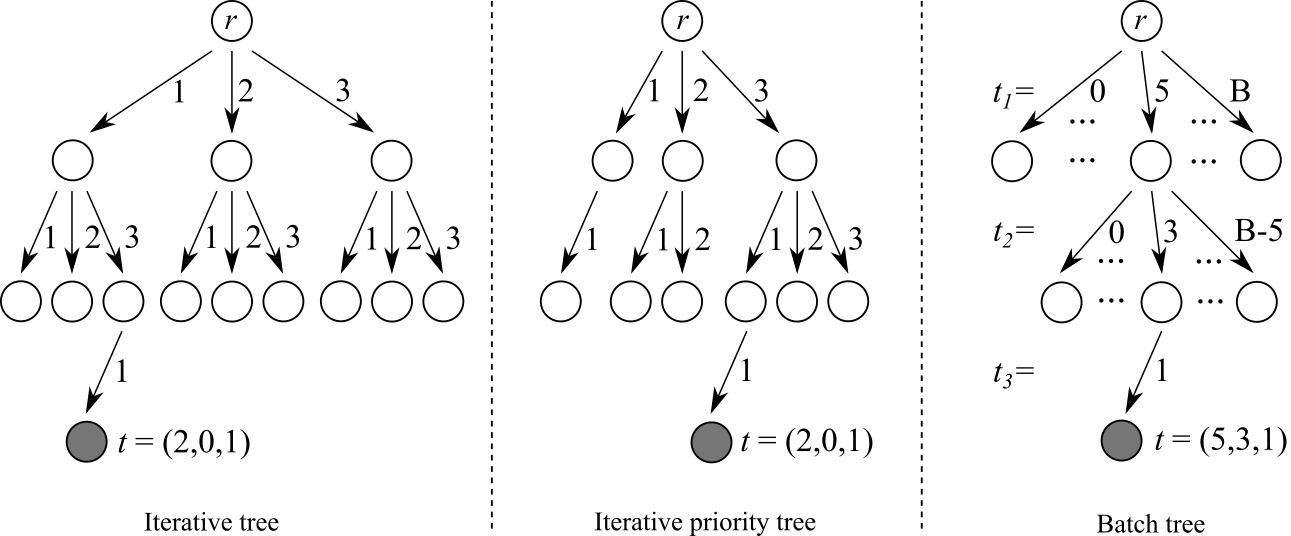}
	\caption{Three tree representations. In the iterative tree (left), the node of the grey dot corresponds to a vector $\bt = (2, 0, 1)$, with its corresponding path expanding hospitals $1$ and $3$ once, and expanding hospital $1$ once again. 
	An edge in the iterative priority tree (middle) is the same as the iterative tree, but its path only includes nonincreasing order of edges (e.g., $1 \rightarrow 3 \rightarrow 1$ is \text{not} allowed).
	An edge in the batch tree (right) corresponds to how many extended seats we allocate to each hospital ($= t_1, t_2, t_3$). The node of the grey dot in the batch tree corresponds to a vector of $\bt = (5,3,1)$. 
	Ordering of the hospitals matters in the iterative priority and the batch trees; the hospitals that receive large expansion should be numbered earliest. All leaves of the iterative and iterative priority trees are at depth $B$, whereas the path length to a leaf of the batch tree varies: It stops branching when $\sum_h t_h = B$.
	%（ここに説明と灰色のノードが具体的にどういう拡張に対応するかを書きます）
	}
	\label{fig_treerep}
\end{figure*}

%適当に名前をつけました
\section{Method}

% permutationでのUCT改良
% https://users.soe.ucsc.edu/~dph/mypubs/AMAFpaperWithRef.pdf
% https://pure.york.ac.uk/portal/files/13014166/CowlingPowleyWhitehouse2012.pdf
% https://web.eecs.umich.edu/~baveja/Papers/fp543-jiang.pdf

% 今回解きたいpermutation (stars and bars problem) 
% https://math.stackexchange.com/questions/58753/distinct-ways-to-keep-n-balls-into-k-boxes/58756

% 偏りのある木ではうまく行かない例

This section proposes our algorithm to solve  Eq.~\eqref{opt_main}. 
We apply Upper Confidence Tree (UCT) \cite{kocsis06}, which belongs to a class of Monte Carlo tree search methods \cite{mctssurvey2012}. One of the most successful examples of UCT is for abstract games such as the Game of Go \cite{WangG07,YoshizoeKKYI11,Silver2016} where the goal is to find the best next move in a game tree. Another aspect of UCT is diverse recommendation method, which minimizes total regret, searches for different
solutions of reasonable objective value~\cite{Bosc2018}.
In this paper, we use UCT to find a global optimum node in a tree $\ETall$. 

UCT consists of a sequential process of traversing a tree structure where each tree node $i$ is associated with its value $v_i$. We use $n\in \{1,2,\dots,N\}$ to denote the number of rounds. The entire tree $\ETall$ is typically very large, and UCT tries to develop a subtree of it. 
We denote the subtree of round $n$ by $\ET(n)$. 
Each node $i$ in the current tree is associated with a tuple $(V_i, N_i) \in \Real^2$, where $V_i$ is the sum of rewards and $N_i$ the number of times at which node $i$ is traversed. The value $\hatmu_i := V_i/N_i$ is an estimator of $v_i$, and its standard deviation is proportional to $\sqrt{1/N_i}$.
Each round consists of \textit{selection}, \textit{development}, \textit{simulation}, and \textit{backpropagation} steps. 

\noindent\textbf{Selection:} At each iteration, UCT traverses the tree from the root. At each node $i$, it chooses a child node $c$ with the maximum UCB value:
\begin{equation}
\UCB(c) := \hatmu_{c} + \Conf \sqrt{\frac{\log(N_i)}{N_c}},
\end{equation}
where $\Conf > 0$ is the parameter that determines how much exploration it attempts.

\noindent\textbf{Development:} 
When it reaches a node $k$ that is out of the current tree, it adds the node to the current tree.%\footnote{Note that we use the term ``development'' in Algorithm \ref{alg_proposed} to distinguish it from the expansion of the hospital seats.}

\noindent\textbf{Simulation:} 
Based on the reached node $k$, it conducts a random play to find $l$, which is a leaf that stems from $k$.

\noindent\textbf{Backpropagation:} 
Letting $v_l$ be the value of the leaf $l$, it updates the statistics of all nodes that we have traversed in this round: $V_i \leftarrow V_i + v_l, N_i \leftarrow N_i + 1$ for each node $i$ during the backpropagation.

The four steps above are quite standard in UCT. 
\ifijcai
For a more algorithmic description, see the full version.
\else
For a more algorithmic description, see Algorithm \ref{alg_proposed} in the appendix.
\fi
%In the following, we describe how to map our capacity expansion problem to UCT.
%Our UCT memorizes the node of the largest value up to round $T$ and output it once it is terminated.

%chooses a child with  expands the search tree $\ETsearch$ by creating a child node, simulates a random traverse from the current node to the leaf node of $\ET$ to evaluate the node, and backpropagate the result of the evaluation to the nodes on the path. 
% developmentのところは普通expansionになっているけど大学の枠と被るので・・・

%もしサーチ部分に工夫を入れた場合はここにAlgorithmを書きます vanilla UCTだと要らないです

% メモ：Pearlの研究（classic)

% ここにMCTSのツリーが今回の問題とどう対応しているのかを書きます
%In this paper, we model...
% ここに工夫点を書きます

%MCTSは各ノード$i$のサンプル数$N_i$と総スコア$V_i$を保持する

\subsection{Tree representation of expansion space}

Unlike abstract games where the game tree is inherent, the relation between a tree and the capacity expansion in our problem is nontrivial. 
%in the case of our capacity expansion problem, 
This section describes the mapping from a tree node into an expansion $\bt \in \Theta$.

In the capacity expansion problem defined in Eq.~\eqref{opt_main}, % the DA algorithm 
DA yields an optimal solution given a fixed expansion vector $\bt$. Therefore, we consider a tree where each node corresponds to an expansion vector $\bt$, and each leaf is an expansion of limit $\sum_{h \in \EH} t_h = B$. The value $v_i$ at each node $i$ is the objective value of % the DA algorithm 
DA for the corresponding expansion.

\begin{table}[b!]
\begin{center}
\caption{Comparison of the three tree representations with respect to the criteria.
}
\label{tbl_criteria}
\scalebox{0.78}{%
    {\renewcommand\arraystretch{1.2}
      \setlength{\tabcolsep}{4pt}
\begin{tabular}{lcccc} \toprule
& Faithfulness & Nonredundancy & Decisiveness \\ \midrule
Iterative tree & %UCT
\checkmark
& 
\\
Iterative priority tree (IPT) & %Pruningと言っているタイプ
\checkmark
& 
\checkmark
\\
Batch tree (BT) & %優先度最大の学校からバッチで割り当てる
\checkmark
& 
\checkmark
&
\checkmark
\\
\bottomrule
\end{tabular}}}
\end{center}
\end{table}

% 良い木の表現に必要な特徴のまとめです
We consider the following criteria are important in constructing a good tree representation of an expansion. 
\begin{enumerate}
\item Faithfulness: The tree preserves the inclusion relationship of the original expansion vector space: if node $j$ is a descendant of node $i$, then $\bt(j) - \bt(i) \ge 0$, where $\bt(i), \bt(j)$ be corresponding expansions of nodes $i,j$ and $\bt(j) - \bt(i) \ge 0$ means all the features of the vector $\bt(j) - \bt(i)$ are non-negative.
% random wireより良い特徴づけとして
\item Nonredundancy: The mapping is one-to-one. Namely, for any expansion vector $\bt \in \Theta$ such that $\sum_h t_h = B$, there exists only one leaf of tree $\ETall$. Redundancy in a tree representation compromises the search efficiency.
\item Decisiveness: Assuming that a tree representation is faithful, the tree representation is decisive if branching in a shallow layer is more informative than the one in a deep layer. In our case, allocations related to important hospitals should appear in a shallow layer of the tree.    
% 重要な決断を早くする?これは後の表現のほうが良いと言えそう
% 例えば[10, 4, 2, 0, 0]のときに大きいほうから決めてく木が最も浅くなる
\end{enumerate}
Considering those criteria, we propose the three tree representations % that we demonstrate 
in Figure~\ref{fig_treerep}.
%shows the structure of these trees.
In all of the three representations, the root node $r$ corresponds to the zero vector $\bt = (0,0,\dots,0)$ that corresponds to no expansion. 
Table~\ref{tbl_criteria} illustrates which properties are satisfied in each of the representations. 
We will discuss the idea behind them. 

The iterative representation, which is the most straightforward, builds an $H$-ary tree.\footnote{Note that $H$ is the number of the hospitals.} 
Each edge corresponds to allocating a seat to one of the hospitals. As a result, each path from the root to depth $B$ represents an expansion of size $B$. Although the iterative representation is faithful, it is redundant. For example, consider the case of three hospitals. Let $1\rightarrow 2 \rightarrow 2$ to denote a path on the tree that sequentially expands hospitals $1,2,2$. An expansion vector $(2, 1, 0)$ corresponds to three different path on the tree ($1\rightarrow 1 \rightarrow 2$, $1\rightarrow 2 \rightarrow 1$, and $2 \rightarrow 1 \rightarrow 1$).
%While we represent capacity expansion as a $H$-ary tree where the order of capacity expansion matters, whereas in our capacity expansion problem, only the expanded capacity $q_h$ for each hospital $h$ matters. 
As a result, UCT with iterative tree representation searches an unnecessary large representation space, which increases the computational burden.  

To deal with this issue, we introduce an improved tree representation of the expansion, which we call the iterative priority tree (IPT). For each node $i$, it only allocates a node with its priority higher than the most prioritized hospital that has been allocated a seat.
%Formally, letting $p(i) \max(\Ind[\bt(i) > 0])$, it only has edge that adds seats for hospital $p$ or higher priority.
IPT solves the issue of redunduncy in the iterative tree:
\begin{prop}{\rm (Nonredundancy of IPT)}\label{prop_prioit_nonred}
For each expansion vector $\bt \in \Natural^\EH$, there exists a unique node in IPT. 
\end{prop}
\begin{proof}
Each expansion uniquely corresponds to a nonincreasing sequence of numbers. For example, expansion $\bt = (1, 3, 1)$ corresponds to $3\rightarrow 2\rightarrow 2\rightarrow 2\rightarrow 1$. It is easy to confirm that, in the iterative priority tree, this is the unique path that leads to the desired expansion.
\end{proof}
% 冗長
%As a result, the number of leaves in the iterative priority representation is ${N\choose k}$
%\[
%\binom{n,x}
%\]
%which corresponds to the number of all allocations of budget $B$.

While the IPT is nonredundant, there still remains some space for improvement. 
%The optimal allocation usually allocates more than one seat (i.e., $t_h > 1$) for some popular hospitals, and IP 
The batch tree (BT) representation fully exploits the priority information to allocate more than one seats to popular nodes, which reduces the depth of the optimal solution in the tree. Namely, each depth of BT corresponds to how many seats to allocate to each hospital. The following proposition states that the properties of BT:
\begin{prop}{\rm (Nonredundancy and decisiveness of BT)}
%Priority-Batchは"1つのedgeで2つ以上の大学に枠を割り振らない"という条件を満たす表現の中で、
%真の割当の数がpriorityの順序と同じ場合、深さが最も短い表現である
%例えば、[3, 2, 1, 0, 0, 0]の場合、batch treeは深さ3で
BT is nonredundant. It is also decisive, that is, if the optimal solution aligns with the expansion vector (i.e., $t_1 \ge t_2 \ge t_3 \ge \dots \ge t_H$), then BT minimizes the depth of the optimal node among all trees such that each edge allocates seats to one hospital.
\end{prop}
\begin{proof}
The nonredundancy is trivial. 
In the following, we derive the decisiveness of the BT.
Assume that the optimal expansion vector is such that  $t_1 \ge t_2 \ge t_3 \ge \dots \ge t_H$, and let $h$ be the first index such that $t_h = 0$. BT includes the path of depth $h-1$ that leads to this allocation.
\end{proof}
% Table \ref{tbl_criteria} summarizes the three tree representations with respect to these criteria.

%結果として、本当の問題の難しさより広い空間を探索してしまい、効率性の低い探索をしていることになる。この問題に対応するためには、
%\begin{itemize}
%\item (1) アップデートの共有：同じ拡張枠のutilityは1度調べれば分かるので、その拡張を導く複数の木のpathに対して更新を行う % Data Sharing, All-Path-Propagation
%\item (2) 枝刈り：愚直なツリーは深さ$B$の$U=|C|$分木になるが、枝及びノードを減らすことによって、ツリーのサイズを小さくする。 %First/Priority Entrance
%\end{itemize}
%の工夫を行う必要がある。In this paper, we propose improvements on these aspects.

%\textbf{Other tree representations:}
%We also describe several other (less effective) improvements of UCT in the supplementary material. 

%\textbf{All-Path-Backpropagation} (APB, Figure \ref{fig:uct_backpropagation}) is an aggressive heuristic in which the reward is propagated not over the tree nodes but over the set relationship in terms of expansion vectors. In APB, we keep $(V_{\vecc}, N_{\vecc})$ for each expansion vector $\vecc \in \Natural^U$. 
% 日本語で書いてもよくわからないですね（要書き直し）
%拡張枠ベクトルから1つの枠を追加することによって別の拡張枠になる関係を辺で表現し、ある枠を更新するに際し、すべてのheadからtailに対しての逆伝搬を行う（Algorithmで書いたほうがわかりやすそう）。 
% tailとhead https://ocw.mit.edu/courses/electrical-engineering-and-computer-science/6-042j-mathematics-for-computer-science-fall-2010/readings/MIT6_042JF10_chap06.pdf

\subsection{Ordering hospitals}

%Although the capacity expansion problem is NP-complete in solving any instance, in many cases there is a way to exploit the structure of the problem. In particular, which hospital should we allocate the extension seats?
We introduce two ideas for ordering hospitals that we use in IPT and BT. 
The first idea is to order hospitals in terms of their popularity
\begin{equation}\label{def_pop}
\mathrm{Popularity}(h) := \sum_{d \in \ED} \rank_d(h).
\end{equation}
We use the term ``popularity'' because it is the total rank of hospital $h$ in view of residents. Note that this value is also referred to as the Borda count in the context of social choice~\cite{moulin:1994}. A smaller value  represents a more popular hospital.
%because each resident gives the lower numbers to the highly ranked hospitals.

The second idea is utilizing \textit{potential envies} that may not be justified so that we could construct a better decision tree of the hospitals. 
% Furthermore, we found an even better way to prioritize the hospitals. 
Namely, let 
\begin{equation}\label{def_envy}
\mathrm{Envy}(h) :=
\sum_{d \in \ED} \Ind[\rank_d(h) < \rank_d(M(d))]
\end{equation}
be the potential envy that the residents have toward the ones matched to hospital $h$ in matching $M$. 
%The potential envy may not be justified because whether  hospital $h$ prefer $d$ to someone matched in $M$ does not matter. 
The score indicates the number of residents who prefer hospital $h$ over their matched hospital in $M$. 
For calculating $M$, we run DA with no expansion before running UCT to determine the ordering.
% 
%We found that this ``envy-based'' order is strongly correlated with the optimal allocation. 
%In our simulation, the correlation coefficient of the envy-based order with the actual order is (value here).

%The first idea, that is inspired by the greedy approach by \cite{bobbio2021capacity}, is to 

%（ここはつなぎ方を変えますーenvyで優先順位が捉えられるが、では具体的なallocationはどうするか）
%Assuming such a priority gives a reasonable way of ordering the hospitals for the extended capacity, the next question is how we can choose the desired allocation: For example, given five hospitals, should we allocate $\bt = (10, 2, 0, 0, 0)$ extended seats for the hospitals, or should we allocate $\bt = (5, 4, 2, 1, 0)$ seats for the hospitals?
%We propose a novel method that combines the aforementioned envy-based ordering with UCT.
%The simulation result shows for all (XXX) instances of size YY, the value of the solution by the proposed algorithm has at least 95\% of the exact method. Moreover, in a large instance of size (YYY) where an exact solution is not obtained within 1-hour with a modern computer, the proposed approach results in better allocation than the temporal solution of the exact algorithm.

%拡張枠：UCTなら固定数以外もできる？
% 地域ごとの制約
% ここに目的関数を書き下します 許容

\paragraph{Avoiding multiple entries.}
\label{subsec_mark}

Since the evaluation of each node (i.e., DA with a given expansion vector) is deterministic, we can restrict the node selection so that a node is never evaluated twice. 
In the backpropagation step, we mark the node $k$ as ``evaluated'' if it is a leaf of $\ETall$. Moreover, for each ascendant node, if all children of the node are evaluated, then we mark that node as evaluated. In the selection and simulation phase, the evaluated nodes are never visited again.
% これはなくてもいいかと
%In a small instance, all nodes in $\ETall$ can be evaluated at round $t<T$.
%In that case, we stop the algorithm.

\begin{table*}[t]
    \centering
    \small
        \caption{Average percentage gaps between the solution found by the Agg-Lin and the solution found by each method for Set~1 experiments. UCT with iterative priority trees has three different orderings, i.e., random, popularity, and envy-based, each of which is denoted by IPT-R, IPT-P, and IPT-E. As well, UCT with batch trees has the three variants: % random/popularity/envy-based ordering are 
        BT-R, BT-P, and BT-E. ``$0.0$'' indicates the value is equal to Agg-Lin.
        Negative values mean the temporal value of Agg-Lin is outperformed. 
    }
    \begin{tabular}{|rrr||r|r||r|r|r|r|r|r|r|} \hline
        \multicolumn{1}{|c}{$H$} & \multicolumn{1}{c}{$B$} & \multicolumn{1}{c||}{$\alpha$} & \multicolumn{2}{c||}{Baseline} & \multicolumn{7}{c|}{\textbf{UCT (proposed)}} \\ \cline{4-12}
        &&& \multicolumn{1}{c|}{LPH} & \multicolumn{1}{c||}{Grdy} & \multicolumn{1}{c|}{Iterative} & \multicolumn{1}{c|}{IPT-R} & \multicolumn{1}{c|}{IPT-P} & \multicolumn{1}{c|}{IPT-E} & \multicolumn{1}{c|}{BT-R} & \multicolumn{1}{c|}{BT-P} & \multicolumn{1}{c|}{BT-E} \\ \hline
        $5$ & $5$ & $0.0$   &  $7.5$ &  $5.7$ &   $0.0$ &   $0.0$ &  $0.0$ &  $0.0$ &  $0.0$ &   $0.0$ &   $0.0$ \\ 
        $5$ & $30$ & $0.0$  &  $6.3$ & $32.9$ &   $1.3$ &  $34.5$ &  $0.0$ &  $0.0$ &  $0.4$ &   $0.0$ &   $0.0$ \\ 
        $15$ & $5$ & $0.0$  &  $8.9$ &  $4.6$ &   $0.5$ &   $0.5$ & $0.05$ & $0.09$ &  $1.1$ &   $1.1$ &   $1.1$ \\ 
        $15$ & $30$ & $0.0$ & $23.2$ & $25.3$ &  $17.8$ &  $18.0$ & $15.5$ & $10.4$ & $19.9$ &  $15.4$ &   $6.9$ \\ \hline
        $5$ & $5$ & $0.2$   &  $1.8$ &  $1.4$ &   $0.0$ &   $0.0$ &  $0.0$ &  $0.0$ &  $0.0$ &   $0.0$ &   $0.0$ \\ 
        $5$ & $30$ & $0.2$  &  $3.6$ &  $6.4$ &   $1.6$ &   $0.9$ &  $0.1$ & $0.03$ &  $0.4$ &  $0.09$ &   $0.1$ \\ 
        $15$ & $5$ & $0.2$  &  $2.6$ &  $0.8$ &  $0.09$ &   $0.0$ &  $0.0$ &  $0.0$ &  $0.2$ &  $0.06$ &  $0.06$ \\ 
        $15$ & $30$ & $0.2$ &  $4.1$ &  $4.3$ &   $2.7$ &   $2.5$ &  $1.3$ &  $1.4$ &  $1.4$ &   $0.3$ &   $0.3$ \\ \hline
        $5$ & $5$ & $0.4$   &  $0.6$ &  $0.2$ &   $0.0$ &   $0.0$ &  $0.0$ &  $0.0$ &  $0.0$ &   $0.0$ &   $0.0$ \\ 
        $5$ & $30$ & $0.4$  &  $1.1$ &  $2.1$ &   $1.4$ &   $0.1$ &  $0.0$ &  $0.0$ &  $0.1$ &   $0.0$ &   $0.0$ \\ 
        $15$ & $5$ & $0.4$  &  $1.1$ &  $0.3$ & $-0.08$ & $-0.09$ & $-0.1$ & $-0.1$ & $0.02$ & $-0.09$ & $-0.09$ \\ 
        $15$ & $30$ & $0.4$ &  $0.8$ &  $0.8$ &   $0.5$ &  $0.05$ & $-0.6$ & $-0.6$ & $0.04$ &  $-0.8$ &  $-0.8$ \\ \hline
    \end{tabular}
\label{tb:average_gap_synthone}
\end{table*}

\begin{table*}[t]
    \centering
    \small
    \caption{Average run times for Set~1 (seconds).}
    \label{tb:run_time_synthone}
    \begin{tabular}{|rrr||r|r|r||r|r|r|r|r|r|r|} \hline
        \multicolumn{1}{|c}{$H$} & \multicolumn{1}{c}{$B$} & \multicolumn{1}{c||}{$\alpha$} & \multicolumn{3}{c||}{Baseline} & \multicolumn{7}{c|}{\textbf{UCT (proposed)}} \\ \cline{4-13}
        &&& \multicolumn{1}{c|}{Agg-Lin} & \multicolumn{1}{c|}{LPH} & \multicolumn{1}{c||}{Grdy} & \multicolumn{1}{c|}{Iterative} & \multicolumn{1}{c|}{IPT-R} & \multicolumn{1}{c|}{IPT-P} & \multicolumn{1}{c|}{IPT-E} & \multicolumn{1}{c|}{BT-R} & \multicolumn{1}{c|}{BT-P} & \multicolumn{1}{c|}{BT-E} \\ \hline
        $5$ & $5$ & $0.0$   &   $24.01$ & $0.01$ & $0.06$ &  $10.09$ &   $0.64$ &   $0.64$ &   $0.65$ &   $0.86$ &   $0.86$ &   $0.86$ \\ 
        $5$ & $30$ & $0.0$  &   $11.85$ & $0.01$ & $0.19$ &  $35.60$ &  $45.38$ &  $39.84$ &  $39.92$ &  $36.93$ &  $38.83$ &  $38.03$ \\ 
        $15$ & $5$ & $0.0$  &  $854.91$ & $0.02$ & $0.94$ &  $52.11$ &  $56.98$ &  $54.88$ &  $54.31$ &  $54.76$ &  $53.50$ &  $53.37$ \\ 
        $15$ & $30$ & $0.0$ &  $362.87$ & $0.02$ & $3.34$ & $177.00$ & $167.24$ & $166.36$ & $160.77$ & $149.91$ & $139.90$ & $139.74$ \\ \hline
        $5$ & $5$ & $0.2$   &   $37.79$ & $0.01$ & $0.09$ &  $15.59$ &   $0.98$ &   $0.99$ &   $0.99$ &   $1.30$ &   $1.31$ &   $1.32$ \\ 
        $5$ & $30$ & $0.2$  &   $55.79$ & $0.01$ & $0.44$ &  $89.59$ &  $94.44$ &  $93.57$ &  $94.46$ &  $91.78$ &  $90.11$ &  $90.06$ \\ 
        $15$ & $5$ & $0.2$  & $2740.02$ & $0.03$ & $1.80$ & $129.87$ & $129.35$ & $129.18$ & $130.38$ & $125.60$ & $128.63$ & $128.88$ \\ 
        $15$ & $30$ & $0.2$ & $3527.88$ & $0.03$ & $7.14$ & $384.11$ & $376.15$ & $370.17$ & $370.66$ & $343.04$ & $362.58$ & $361.80$ \\ \hline
        $5$ & $5$ & $0.4$   &   $86.30$ & $0.01$ & $0.09$ &  $15.04$ &   $0.98$ &   $0.97$ &   $0.97$ &   $1.30$ &   $1.29$ &   $1.30$ \\ 
        $5$ & $30$ & $0.4$  &  $294.07$ & $0.01$ & $0.54$ & $123.79$ & $130.87$ & $129.49$ & $129.18$ & $123.63$ & $123.91$ & $123.89$ \\ 
        $15$ & $5$ & $0.4$  & $3600.21$ & $0.03$ & $2.38$ & $169.40$ & $166.00$ & $166.96$ & $168.14$ & $168.51$ & $166.51$ & $166.11$ \\ 
        $15$ & $30$ & $0.4$ & $3600.41$ & $0.04$ & $9.24$ & $578.57$ & $606.01$ & $574.18$ & $584.77$ & $553.69$ & $571.35$ & $572.87$ \\ \hline
    \end{tabular}
\end{table*}

\begin{table*}[t]
    \centering
    \small
    \caption{Average percentage gaps between the solution found by the Agg-Lin and the solution found by each method for JRMP.}
    \label{tb:average_gap_jrmp}
    \begin{tabular}{|r||r|r||r|r|r|r|r|r|r|} \hline
        \multicolumn{1}{|c||}{$B$} & \multicolumn{2}{c||}{Baseline} & \multicolumn{7}{c|}{\textbf{UCT (proposed)}} \\ \cline{2-10}
        & \multicolumn{1}{c|}{LPH} & \multicolumn{1}{c||}{Grdy} & \multicolumn{1}{c|}{Iterative} & \multicolumn{1}{c|}{IPT-R} & \multicolumn{1}{c|}{IPT-P} & \multicolumn{1}{c|}{IPT-E} & \multicolumn{1}{c|}{BT-R} & \multicolumn{1}{c|}{BT-P} & \multicolumn{1}{c|}{BT-E} \\ \hline
        $10$ & $0.2$ & $0.3$ & $1.3$ & $1.1$ & $0.04$ & $0.06$ & $0.9$ & $0.02$ & $0.02$ \\ 
        $30$ & $0.4$ & $1.1$ & $4.8$ & $3.7$ &  $0.1$ &  $0.4$ & $1.7$ & $0.02$ & $0.09$ \\ \hline
    \end{tabular}
\end{table*}
% \end{landscape}

\begin{table*}[t]
    \centering
    \small
    \caption{Average run times for JRMP (seconds).}
    \label{tb:run_time_jrmp}
    \begin{tabular}{|r||r|r|r||r|r|r|r|r|r|r|} \hline
        \multicolumn{1}{|c||}{$B$} & \multicolumn{3}{c||}{Baseline} & \multicolumn{7}{c|}{\textbf{UCT (proposed)}} \\ \cline{2-11}
        & \multicolumn{1}{c|}{Agg-Lin} & \multicolumn{1}{c|}{LPH} & \multicolumn{1}{c||}{Grdy} & \multicolumn{1}{c|}{Iterative} & \multicolumn{1}{c|}{IPT-R} & \multicolumn{1}{c|}{IPT-P} & \multicolumn{1}{c|}{IPT-E} & \multicolumn{1}{c|}{BT-R} & \multicolumn{1}{c|}{BT-P} & \multicolumn{1}{c|}{BT-E} \\ \hline
        $10$ & $247.89$ & $0.08$ &  $9.35$ & $20.42$ & $22.11$ & $21.49$ & $22.08$ & $20.58$ & $21.13$ & $21.16$ \\ 
        $30$ & $438.20$ & $0.08$ & $27.12$ & $62.72$ & $74.68$ & $73.18$ & $73.65$ & $77.30$ & $74.98$ & $78.20$ \\ \hline
    \end{tabular}
\end{table*}

\begin{figure*}[t!]
    \centering
    \begin{minipage}[t]{0.32\textwidth}
        \centering
	    \includegraphics[width=1.1\textwidth]{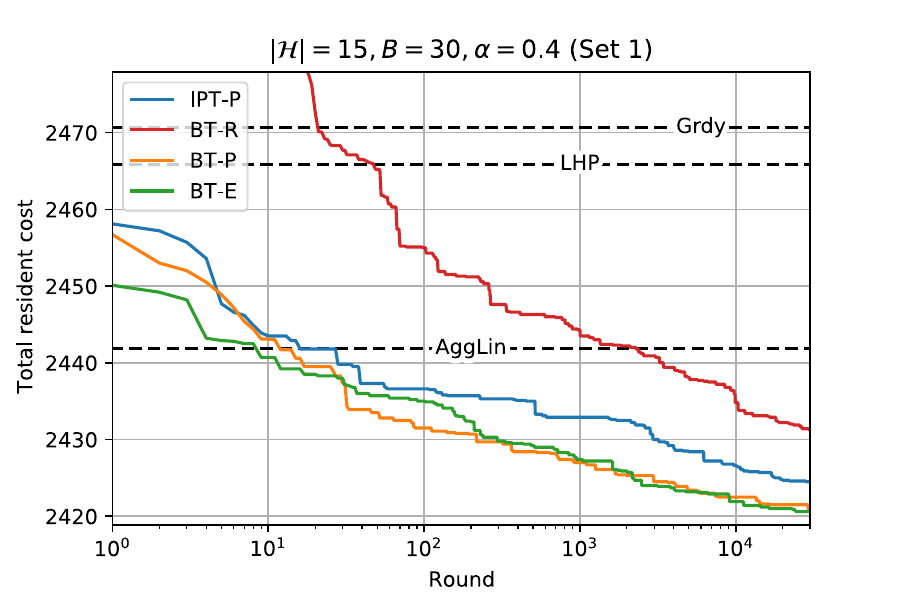}
    \end{minipage}
    \begin{minipage}[t]{0.32\textwidth}
        \centering
	    \includegraphics[width=1.1\textwidth]{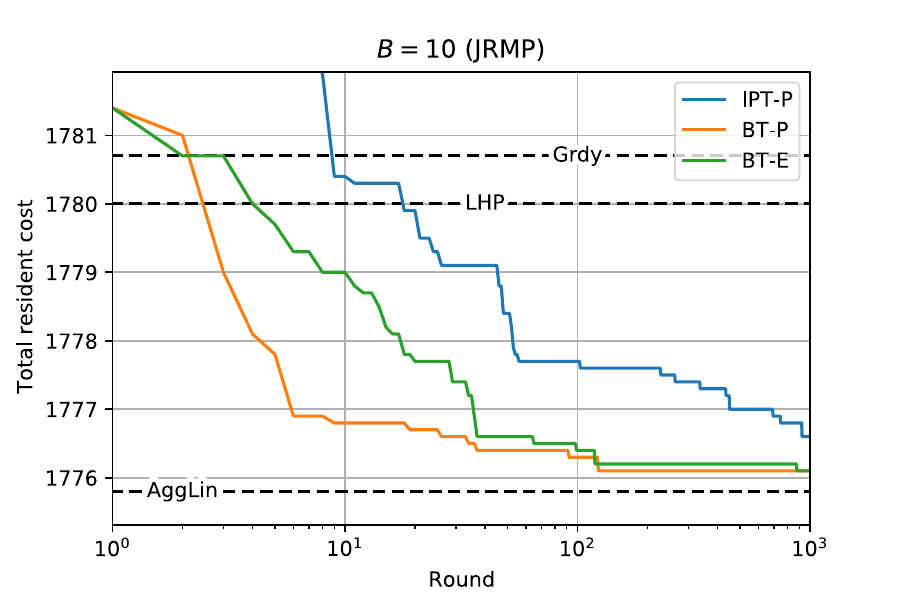}
    \end{minipage}
    \begin{minipage}[t]{0.32\textwidth}
        \centering
	    \includegraphics[width=1.1\textwidth]{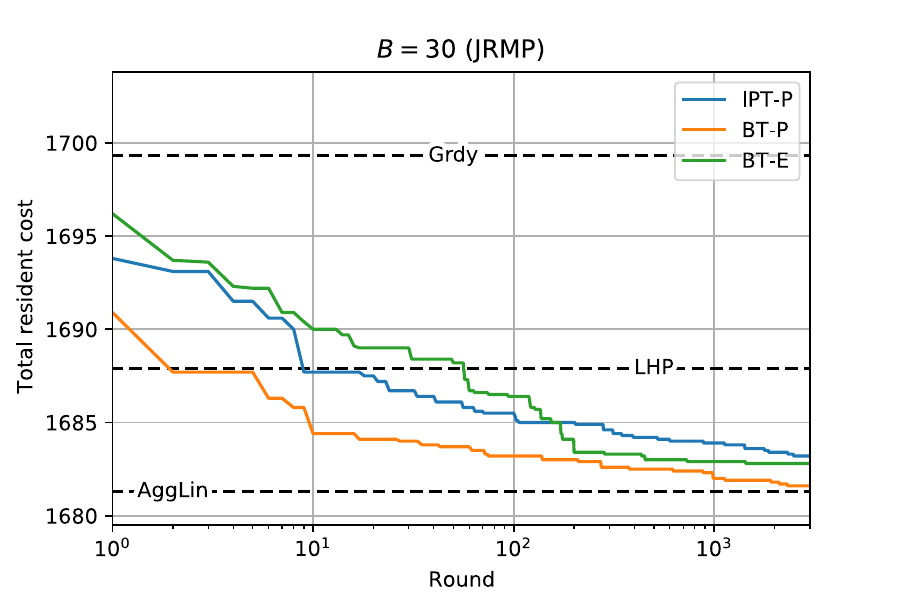}
    \end{minipage}
    \caption{Total resident rankings (costs) of the proposed algorithms. Horizontal axis indicates the number of rounds. % $n$. 
    The dotted lines represent the total costs obtained by Grdy, LPH, and Agg-Lin.}
    \label{fig:trajectory}
\end{figure*}

%解析
\subsection{Consistency of the method}

This section analyzes the proposed method. The first result is a trivial consequence of avoiding multiple entries: % (Section \ref{subsec_mark}):
\begin{prop}{\rm (Worst-case sample complexity)}
\label{prop_worst}
UCT with iterative tree finds an optimal solution in $N = |\ETall|$ rounds.
%UCT with iterative tree finds an optimal solution in $N = (H^{B+1}-1)/(H-1)$ rounds. UCT with iterative priority tree or batch tree finds an optimal solution in $N = |\ETall|$ rounds.
%$N = \sum_{b\le B} {H+b-1 \choose b-1}$ rounds.
\end{prop}
\begin{proof}
Since the development step adds a node to $\ET(n)$ at each round $n$, it follows that it obtains the optimal solution if $N$ reaches the number of the nodes $|\ETall|$.
%, which is $\sum_{i\le B} n^i$ in the iterative tree or the number of possible expansions of limit $\sum t_h \le B$ in the iterative priority or the batch tree.
\end{proof}

%While Proposition \ref{prop_worst} gives a consistency of the proposed method, it does not exploit the property of the confidence bound. 
%Since our UCT never evaluates the same leaf twice, it follows that it obtains the optimal solution if $T$ exceeds the number of all the expansion vector of capacity $B$ (i.e., ${B \choose C+B-1}$). 
%The NP-completeness of this problem (Proposition \ref{prop_hardness}) implies that the polynomial-time solution to $C$ is unavailable,\footnote{Unless $P = NP$.}. 

%だいたい書けたので数日中に書き直します
We further propose a nontrivial analysis by following the analysis of the original UCT by \cite{kocsis06}. The largest difference is that \citeauthor{kocsis06}~\shortcite{kocsis06} analyzed the average quality of the selected action (i.e., regret), whereas we are interested in the quality of the best path, which corresponds to the optimal solution of the capacity expansion problem.

Let $\hatmu_{i,m}=V_i/m$ be the mean of $v_i$ over the first $m$ visits, and $U_{i,m,n} = \hatmu_{i,m} + \Conf\sqrt{\frac{\log n}{m}}$ be the corresponding UCB value. Let $W$ be the maximum number of the children of a node and $D_T$ be the depth of the tree.\footnote{That is $W=H, D_T=B$ for iterative and iterative priority tree, or $W=B,D_T=H$ for the batch tree.} 
Let $v^*_i$ be the maximum value of the leaves that are descendants of node $i$. Let $v^* = v^*_r$ be the global optimal value, where $r$ is the root node. The following assumption is essentially the same as the one in \cite{kocsis06}:
\begin{assp}{\rm (Confidence bound)}
\label{assp_confbound}
We assume % the following:
\begin{equation}\label{ineq_confbound}
\Prob\left[ U_{i,n,m} \ge v^*_i \right] \ge 1 - 1/n^2.
\end{equation}
\end{assp}
\begin{definition}{(\rm $\Delta$-optimal tree)}
Let $\Delta > 0$ be arbitrary. The $\Delta$-optimal subtree is defined recursively as follows. First, it includes the root of the original tree. For each node $i$ in the subtree, add each children $c$ such that $v_i^* - v_c^* \le \Delta$. In other words, $\Delta$-optimal subtree is a subtree where each edge is suboptimal at most $\Delta$.
Let $S(\Delta)$ be the number of the nodes in the $\Delta$-optimal subtree.
\end{definition}

\begin{thm}\label{thm_complexity}
If Assumption \ref{assp_confbound} holds for any node $i$, then with probability at least $1 - o(N^{-1})$, UCT finds at least one node in the $\Delta$-suboptimal tree if $N$ satisfies
\begin{equation}\label{ineq_suffT}
N > W S(\Delta)
\frac{(\Conf)^2 \log N}{\Delta^2}.
\end{equation}
\end{thm}
\ifijcai
The proof of Theorem~\ref{thm_complexity} is in the full version.
\else
The proof of Theorem~\ref{thm_complexity} is in Appendix~\ref{sec:Thorem 5}.
\fi

\begin{remark}{\rm (Implication)}
%Since the depth of the tree is at most $B$, 
The solution of the node $i$ in $\Delta$-suboptimal tree is at least
\[
v_i \ge v^* - D_T \Delta.
\]
Therefore, letting $\Delta$ be sufficiently small and $N$ be sufficiently large, UCT finds an almost optimal solution. The required number of $N$ depends on $S(\Delta)$, which corresponds to the number of close-to-optimal children.
%A smaller value of $\Delta$ leads to a smaller $S(\Delta)$ at the price of diverging $1/\Delta^2$ term. 
%Theorem \ref{thm_complexity} holds for any $\Delta > 0$ we have 
%The value $S(\Delta)$ is a generally large quantity, which we consider unavoidable given the hardness of the problem.
%Although the assumption of $\Delta$-suboptimality (Definition \ref{def_deltasub}) is rather impractical, it captures the essence of UCT. At each node, UCT draws each suboptimal arm for $O(\log T/\Delta^2)$ times, where $\Delta$ is the reward gap between the best action and a suboptimal action.
%Theorem \ref{thm_complexity} also implies the importance of the decisiveness property: The term $\prod_{d=1}^D \Beff_d$ is exponential to the depth $D$ of the optimal node.
%We also note that the definition of $\Delta$ is practically justified in some contexts. For example, \cite{imagawa2015} demonstrated that a suboptimal node tends to have many descendants of low values, and exploiting such property can improve the efficiency of a tree search. 
\end{remark}

\ifijcai
\else
Although our analysis gives a complexity bound based on the structure of the tree, we still think this analysis (as well as any existing analysis of UCT) is not very satisfying. 
We discuss the limitation of this analysis in Appendix~\ref{subsec_analysis_rev}.
\fi
 
%（ここに関連論文のレビューを書きます）
%\begin{remark}{\rm (Review of related results)}
%aaa
%Most of the 
%\cite{KaufmannK17}はノードが展開済のupper confidence treeの解析を行った。

%The analyses above have two major limitations: Namely,
%\begin{itemize}
%\item ノードが展開済であることを仮定しており、真のUCTとは異なる
%\item Game of GoのようなMax-min木を探索している。Max-Min木では安全な枝刈りができる可能性があるが、今回の問題はほとんどのleafがダメなノードからでも最適解がある可能性がある。
%\end{itemize}
%\end{remark}

\section{Evaluation}
This section empirically evaluates our algorithm via synthetic and real datasets. 
%
%以下のアルゴリズムを比較した. %GreedyとExactは\cite{bobbio2021capacity}によって提案された手法である.。。
%\begin{itemize}
%\item DA: 拡張しない方法（DAを$B=0$で回す）
%\item Greedy:
%\item UCT系
%\item Exact（McCormick Envelopeを使用した方法としていない方法）
%\end{itemize}
% 
%All results are averaged over $\Runnum{}$ runs.
%実験結果：
%\begin{enumerate}
%\item 各アルゴリズムの性能比較および実行時間比較を行う
%\item UCT系アルゴリズムの解のラウンド数に対する関数を表示する
%\item UCT系アルゴリズムの解の実行時間に対する関数を表示する
%\end{enumerate}
% 
%考察
%\begin{enumerate}
%\item Greedy / UCT（提案手法群） / Exactの比較
%\item 解となる拡張枠の性質 Bのみのときと$B_u$ごとのとき
%\item 人気の大学に枠を割り振るのがどの程度最適解になりやすいか
%\end{enumerate}
% 
% \subsection{Compared algorithms}
We compare our algorithms %  and the variants 
with the ones by~\cite{bobbio2021capacity}: the greedy algorithm (Grdy), the linear programming-based heuristic (LPH), and the aggregated linearization (Agg-Lin). 
%Note that, throughout this section, we follow the description that the objective of their formulation minimizes the total resident cost. 
% , which is equivalent to maximizing the total resident utility.

Let us enumerate existing algorithms: First, Grdy is an algorithm that allocates $B$ expansion seats iteratively such that each allocation maximizes the marginal cost reduction. 
Second, LPH first computes a minimum cost matching without stability constraints via minimum-cost flow to fix the expansion seats, and then uses DA to obtain a stable matching.
Finally, Agg-Lin is an exact method that uses McCormick envelopes and solves a mixed integer programming, which is computationally intensive.
\ifijcai
\else
\footnote{See Appendix~\ref{exact_and_uct} for the comparison between general methods for solving mixed integer programming and our UCT.}
\fi
We warm-up Agg-Lin as suggested by \cite{bobbio2021capacity}. 
%はMcCormick envelopを利用して線形化を行う手法で、最悪計算量がNP-Completeであり正確な解を得ることができる点はIQPと同様だが、実用上IQPより短い時間で終了する傾向がある。

For our UCT method, we set $N = B \times 10^3$ in synthetic data experiments, $N = B \times 10^2$ in real data experiments.
The value of $\Conf$, which determines the tradeoff between exploration and exploitation, is set to be $\sqrt{0.002}$. We consider this parameter is robust enough to cover all settings. 
We implement our simulation with Python 3 and solve the mathematical programming using Gurobi, which is in favor of Agg-Lin.
% Our program is implemented by Python 3. Integer programming and mixed-integer linear programming problems are solved by Gurobi, which is in favor of Agg-Lin.
We restrict the runtime for each method to one hour.

% メモ：B_uがある設定とない設定
% \subsection{Synthetic and real data}
Let us describe the procedure of generating the datasets. % for our evaluation. 
First, we build two synthetic datasets. \textbf{Set~1} involves $B$ but each hospital does \textit{not} have its expansion limit $b_h$. \textbf{Set~2} involves $B$ as well as $b_h$ for each hospital. Note that the setting of Set~1 is similar to the experimental section in \cite{bobbio2021capacity}. Regarding the preference among the residents, we follow the setting \cite{Goto:aij:2016}, which involves a correlation parameter $\alpha \ge 0$. Larger value of $\alpha$ implies a stronger correlation among the preferences. Note that $\alpha=0$ (no correlation) corresponds to the setting of \cite{bobbio2021capacity}.
We set $D := |\ED|=1,000$, and conduct experiments varying the parameters $H := |\EH|$, $B$, $b_h$, and $\alpha$.
For each combination of parameters, we average the results for $10$ instances.
\ifijcai
%The details of the data generation are in the full version.
\else
The details of the data generation are in Appendix~\ref{sec_synth_dgp}.
The limitation of this analysis in Appendix~\ref{subsec_analysis_rev}.
\fi

% \subsection{Real data}

%（ここに実データでの設定を書きますーページ数次第ですが、たぶんsupplementary materialに行くと思います）
% kamada論文のデータ記述のsummaryです
Second, we generate % the performance of the proposed method in 
the dataset, \textbf{JRMP}, based on Japan Residency Matching Program % (JRMP)
2007~\cite{kamada:aer:2015}, which matches medical hospital students (residents) with residency training programs. % For ease of discussion, we call each training program ``hospital''. 
We extracted $1,287$ residents in the Tokyo district who match with $50+1$ ($1$ for a dummy) hospitals with a resident-side preference. We set $B \in \{10, 30\}$ for the limit of capacity expansion.
For each $B$, we average the results for $10$ instances.
\ifijcai
\else
We place the details of the dataset in Appendix \ref{sec_real_dgp}. 
\fi

% \subsection{Results}
Tables~\ref{tb:average_gap_synthone} and \ref{tb:run_time_synthone} illustrate the quality of solutions and its runtime for Set~1, respectively. 
\ifijcai
The results for Set~2, which is omitted due to page limitation, are found in the full version.
\else
We place the results for Set~2 in Appendix~\ref{results_settwo}.
\fi
Also, we place  Tables~\ref{tb:average_gap_jrmp} and \ref{tb:run_time_jrmp} for the JRMP data. 
Note that Table \ref{tb:average_gap_synthone} indicates the average percentage gap in the total resident ranking (TRR) defined in Eq.~\eqref{eq_utility}: 
\[
100\times \frac{
\text{(TRR of the method)} 
-
\text{(TRR of Agg-Lin)} 
}{
\text{(TRR of the method)} 
}.
\]

% , which is defined as 
% \[
% 100\times \frac{
% \text{(Eq.~\eqref{eq_utility} of the method)} 
% -
% \text{(Eq.~\eqref{eq_utility} of Agg-Lin)} 
% }{
% \text{(Eq.~\eqref{eq_utility} of the method)} 
% }.
% \]
Agg-Lin always outputs an optimal solution upon the completion. However, as in Table~\ref{tb:run_time_synthone}, it does not run within one hour for large instances with positive correlation $\alpha$. In that case, we use the temporal value that Gurobi outputs, which can be suboptimal and the percentage gap can be negative.

Among the algorithms we consider, the two greedy methods (Grdy and LPH) run very fast ($< 10$ seconds) and output a suboptimal solution. 
Both of them are outperformed by all variants of ours.
In particular, BT outperformed iterative tree and IPT, and 
using the popularity and envy-based orderings outperformed the random one. 
%In small instances (e.g., $(|\mathcal{H}|,B,\alpha) = (5,5,0.0)$), all of them outputs an optimal solution. In large instances, the batch tree performs best, followed by the iterative priority tree and the iterative tree. 

% \textbf{Objective as a function of round:} 
Figure \ref{fig:trajectory} describes the objective value as a function of rounds, which indicates an early termination of BT often has a satisfying solution. %This implies the possibility of further reducing the computational time of UCT.
%Among the two orderings, envy-based ordering performs outperformed the popularity-based ordering. 
%The envy-batch algorithm, which performs best in overall among the proposed algorithms. 
%In some instances, Agg-Lin does not finish within one hour. In particular, Agg-Lin tends to be slow in the instances with $\alpha = 0.4$, and the temporal solutions of Agg-Lin for these instances are suboptimal.
In summary, BT with envy/popularity-based ordering yields the best results among our algorithms. Our algorithms run 2--20 times faster than Agg-Lin, and the quality of the solution is close to optimal.

%実データの結果
%Tables \ref{tb:average_gap_jrmp} and \ref{tb:run_time_jrmp} show the results with the real data. UCT algorithms run 5--6 times faster than Agg-Lin, and output an almost optimal solution. Similar to the synthetic dataset, the batch tree outperformed the other two trees. One thing that is worth mentioning is that the popularity-based ordering outperformed the envy-based ordering.

%\begin{itemize}
%\item Greedyは一瞬で終わるがsuboptimalである度合いが大きい
%\item IPQ, Agg-Linはサイズが大きいと1時間でおわらないケースがある\footnote{This observation is consistent with \cite{bobbio2021capacity}.}
%\item AMAF（名前変更予定）はappendixで説明する毛色の異なる改良だが、UCTよりはっきりと性能がよい。サイズの小さいインスタンスだと遅い傾向にある。
%\end{itemize}

%メモ：UCT系のアルゴリズムの速さは何で決まるのか？同じラウンドならDA以外の部分ってそんなに時間かからない気が...

\section{Conclusion}

This paper sheds light on the capacity expansion in the two-sided matching to consider a flexible allocation of extra seats within a given limit beforehand.
To handle this NP-Complete problem, we develop a UCT-based search method and verify that it outperforms the previous approaches. % , including the exact one. 
Future works % Interesting directions of future work ググったら4,160しかヒットしない．．．
include (1) extending it to matchings with constraints that need not admit stable matching and (2) utilizing other tree search algorithms such as Nested Monte-Carlo Search 
\cite{DBLP:conf/ijcai/Cazenave09,DBLP:conf/ijcai/Rosin11}. % to the capacity expansion.

\clearpage

\bibliographystyle{named}
\bibliography{references_ijcai22,approxmatching}

\clearpage

\appendix

\section{Comparison Between Branch-and-bound and UCT}
\label{exact_and_uct}

The aggregated linearization method \cite{bobbio2021capacity} linearlizes the original objective (i.e., Eq.~\eqref{opt_main}) and uses an off-the-shelf solver of mixed integer programming. The branch-and-bound (BB) method is one of the most promising methods for the mixed integer programming.\footnote{Note that BB is also adapted by Gurobi, which combines BB with many efficient heuristics  \url{https://www.gurobi.com/resource/mip-basics/}} 
BB for a minimization problem involves a lower bound (= solution of the linear relaxation) and an upper bound (= the current best feasible solution). BB and UCT both utilize a tree structure and are an iterative process of refining a solution. Here are the differences between BB and UCT.
\begin{itemize}
\item Branch-and-bound, applied to our problem, searches the (linearlized version of) joint space of $(\bx, \bt) \in \{0,1\}^\EE \times \Theta$ whereas UCT only searches for the space of $\bt \in \Theta$, which is significantly smaller than the joint space. UCT exploits the fact that DA optimizes $\bx$ given $\bt$.
\item For obtaining a feasible solution, general solvers adopt general heuristics that are ignorant of the problem structure of capacity expansion, whereas UCT exploits the structure of the problem (i.e., the priority of the hospitals and good tree representation on $\Theta$). Giving a good feasible solution enables a more aggressive branching, and we may combine UCT with branch-and-bound in that sense.
\end{itemize}

In summary, UCT exploits the structure of the problem unlike general solvers with BB. 
In our simulation, we demonstrate UCT provides a better solution than the intermediate solution of the exact method for large instances with a shorter amount of runtime, even though the implementation of UCT in Python is in favor of the exact method.

\section{Algorithmic Description of UCT}

Algorithm \ref{alg_proposed} describes the steps in our UCT method. UCT can be viewed as a combination of multi-armed bandit algorithm and monte-carlo search. 

\begin{algorithm}[!t]
\caption{Upper Confidence Tree for Searching Capacity Expansion}
\label{alg_proposed}
\begin{algorithmic}
    \REQUIRE \# of Rounds $N$. %, Limit $B$.%, $C$: \# of hospitals
    \STATE Initialize the tree with the root node $\ET(1) = \{r\}$.
    \FOR{$n=1,2,\dots,N$}
        \STATE Set the current node $i$ to the root node: $r$.
        \WHILE{Current node $i$ is in $\ET(n)$}
            \STATE Find the most promising child $c$ and set the current node to $c$, where
            \begin{equation}\label{ineq_ucb}
                \argmax_{c'} \UCB(c').
            \end{equation} \COMMENT{Selection}
        \ENDWHILE
        \STATE $k \leftarrow i$.
        \STATE Add the current node $k$ to the tree $\ET(n+1) \leftarrow \ET(n) \cup \{k\}$ and initialize the related statistics $(V_{k}, N_{k}) = (0, 0)$. \COMMENT{Development}
        \STATE Randomly select a leaf $l$ of $\ETall$, which is one of the descendants of the current node $k$. Evaluate $v_l$ that is the result of DA with the corresponding expansion. \COMMENT{Simulation} 
        \STATE Backpropagate value $v_l$ from $k$ up to the root node: For each node $i$ between $k$ and $r$, it updates 
        $V_i \leftarrow V_i + v_l, N_i \leftarrow N_i + 1$.\COMMENT{Backpropagation}
    \ENDFOR
\end{algorithmic}
\end{algorithm}
%今回解きたい問題のstate（expension）はツリーの構造からの同相とか書きます
%囲碁などの展開型ゲームではツリー構造と局面の状態が1:1対応しているため、AMAFなどの複数の局面の状態を同一視する方法はstate abstractionと見ることができるが、今回の問題は拡張枠はツリー構造からの同相写像であるため、ツリーの複数の局面を同一視することによって初めて正確なstateの表現が得られる.

% 今川さんの博論（papersフォルダ内）はUCTのことを最良優先探索と言っている
% https://arxiv.org/pdf/1706.02986.pdf Kaufmann
% https://arxiv.org/pdf/1206.3382.pdf BLUE, 直接は使わないけど参考になる（書き方など）
% http://www.incompleteideas.net/609%20dropbox/other%20readings%20and%20resources/MCTS-survey.pdf 引用の多いサーベイ
% MCTSではexpandする前のtreeとexpandした全体のtreeをどう書くのかと思って関連論文を見ていたのですが, 
% ↑のsurveyだとアルゴリズムが持っているsearch treeと言っているのでそう書きます

\section{Proof of Theorem \ref{thm_complexity}}
\label{sec:Thorem 5}

%証明はappendixに行く予定です
\begin{proof}
Let $N' = N - 1$.
We show that, with a high probability, any node that is \textit{not} in the $\Delta$-optimal subtree is visited as most
$\max\left(
\frac{N'}{W S(\Delta)}, \frac{(\Conf)^2 \log N}{\Delta^2}
\right)
$ rounds. 
This implies that UCT spends at most $\max(N-1, W S(\Delta)
\frac{(\Conf)^2 \log N}{\Delta^2})
$ rounds\footnote{The number of the nodes that has an  incoming edge from a node of $\Delta$-optimal subtree is at most $W S(\Delta)$.} in a path that includes a node outside the $\Delta$-optimal subtree; from which it visits a path in the $\Delta$-optimal subtree if 
\[
N > 
W S(\Delta) \frac{(\Conf)^2 \log N}{\Delta^2}
\]
rounds.

Assumption \ref{assp_confbound} implies that the event
\begin{equation}\label{ineq_unifconf}
\bigcup_{
n \ge \max\left(
\frac{N'}{W S(\Delta)}, \frac{(\Conf)^2 \log N}{\Delta^2}
\right)
}
\left\{
U_{i,n,N_c(n)} \ge v^*_i
\right\}
\end{equation}
occurs with probability at most
\[
\sum_{n \ge \frac{N'}{WS(\Delta)}}^N
\frac{1}{n^2} \le 
N \left(\frac{WS(\Delta)}{N'}\right)^2
%\le 
%\frac{(WS(\Delta))^2}{N'}
\le 
\frac{2(WS(\Delta))^2}{N}
\]
and its union bound over $S(\Delta)$ nodes occurs with probability at most $2W^2(S(\Delta))^3 N^{-1}$. In the follows, we assume Eq.~\eqref{ineq_unifconf} for any node in the $\Delta$-suboptimal tree. 

Let $i$ be an arbitrary node in the $\Delta$-suboptimal tree and $n$ be an arbitrary round. Suppose that UCT visits a node $j$ that is not in the $\Delta$-suboptimal tree for more than 
\begin{equation}\label{ineq_njmin}
N_j(n) > \max\left(
\frac{N'}{W S(\Delta)}, \frac{(\Conf)^2 \log N}{\Delta^2}
\right)
\end{equation}
times. 
We have 
\begin{align}
U_{j,n,N_j(n)}
&\ge 
U_{i,n,N_c(n)}\\
&\ge v^*_i \text{\ \ \ (by Ineq.~\eqref{ineq_unifconf})}
%&= v^*.
\end{align}
Moreover, by definition of the UCB value, we have
\begin{align}
U_{j,n,N_j(n)} 
&:=
\hatmu_{j,N_j(n)} + \Conf\sqrt{\frac{\log n}{N_j(n)}}\\
&\le 
v_j^* + \Conf\sqrt{\frac{\log n}{N_j(n)}}\\
&\le 
v_j^* + \Conf\sqrt{\frac{\log N}{N_j(n)}}\\
\end{align}
Combining these two equations yields
\begin{equation}
N_j(n)
\le
\frac{(\Conf)^2 \log N}{(\mu_i^*-\mu_j^*)^2}
\le
\frac{(\Conf)^2 \log N}{\Delta^2}
\end{equation}
which contradicts with Eq.~\eqref{ineq_njmin}.
By contradiction, node $j$ is never visited again after Eq.~\eqref{ineq_njmin} is satisfied.

In summary, with probability at least $1-2W^2(S(\Delta))^3 N^{-1}$, Eq.~\eqref{ineq_unifconf} holds for any node in the $\Delta$-suboptimal tree. Under Eq.~\eqref{ineq_unifconf}, any node that is not in the $\Delta$-optimal subtree is visited as most
$\max\left(
\frac{N'}{W S(\Delta)}, \frac{(\Conf)^2 \log N}{\Delta^2}
\right)
$ rounds, from which it visits at least one path in the $\Delta$-optimal subtree if 
\[
N > 
W S(\Delta) \frac{(\Conf)^2 \log N}{\Delta^2}.
\]
\end{proof}

\section{Comparison of Existing Theoretical Results}
\label{subsec_analysis_rev}

%We avoid an MDP-based formulation of \cite{kocsis06} that blurs the complexity with respect to the tree size, which we consider to be a clear cut.
Assumption \ref{assp_confbound}, which is essentially the same as the one in the literature,\footnote{Eq.~(3)--(4) in \cite{kocsis06}} is very strong and is not satisfied unless the tree is fully developed.\footnote{Here, we use the word ``development'' to denote the expansion of the upper confidence tree so that it is not confused with the capacity expansion.} We are not sure if any analysis of UCT gets rid of this limitation. Still, given the overwhelming utility of the UCT method in many practical domains, we consider the value of UCT cannot be overstated. 

Recently, there are more solid analyses that rely on less restrictive conditions. However, these analyses tend to modify the original UCT and do not apply to our case. One of the seminal papers by \cite{zohar2014} analyzed a variant of UCT called BLUE. To get rid of stringent assumptions, BLUE considers a two-phase strategy that consists of the development and evaluation phases. Another notable work by \cite{KaufmannK17} brought a solid analysis on a fully developed game tree.

Note that the game tree search is inherently different from our problem. In a two-player zero-sum game, the best move for the black player is the worst move of the white player, and thus the value is flipped at each edge, which enables safe pruning of the tree (i.e., alpha-beta cuts \cite{alphabeta}, or its confidence-based pruning in \cite{KaufmannK17}).
%For example, a solid analysis by \cite{KaufmannK17} exploits this property. However, such a pruning does not exist in the expansion space. In the worst case, we need to explore all nodes unless we exploit some properties of the DA algorithm, which is very hard given some corner-case instances (c.f., counterexamples introduced in \cite{bobbio2021capacity}).

\section{Details of simulations}

%とりあえずsubsectionにしますが細かすぎたら後で戻します
\subsection{Computational environments:}
Our program is implemented in Python 3. Linear programming problems and mixed integer programming problems are solved by the Gurobi optimizer. We consider this setting is in favor of the aggregated linearization method, which fully utilizes the power of the state-of-the-art optimizer, whereas Python implementation has some space for improvement (e.g., reimplementation via low-level programming languages).  
All simulations are conducted on the Google cloud platform (e2-standard-8 instance, eight-core, 32GB memory). All programs are single-threaded. 
%（メモ：DAの実装もPython3？）
%(Intel(R) Xeon(R) CPU on 2.20GHz, Linux 10)

%\paragraph{Synthetic dataset.} 
%We consider two sets of experiments. \textbf{Set~1} involves $B$ but each hospital does \textit{not} have its expansion quota $b_j$. \textbf{Set~2} involves not only $B$ but also $b_j$ for each hospital. Note that the setting of Set~1 is similar to the experimental section in \cite{bobbio2021capacity}.
%\begin{itemize}
%\item[Case 1] This case involves the total expansion seat $B$. Note that this setting is the same as \cite{bobbio2021capacity}.
%\item[Case 2] This case involves 
%\item[Case 3] This sca
%\end{itemize}

\subsection{Details of synthetic data}
\label{sec_synth_dgp}

We build two sets of experiments. \textbf{Set~1} involves $B$ but each hospital does \textit{not} have its expansion quota $b_h$. \textbf{Set~2} involves $B$ as well as $b_h$ for each hospital $h$. Note that the setting of Set~1 is similar to the experimental section in \cite{bobbio2021capacity}. 

We generate preference lists for hospitals and capacities uniformly at random satisfying the following provisions: (1) no hospital has capacity zero; (2) $\sum_{h\in \EH} q_h = D$.
Preference lists of residents are generated by the following procedure: (1) generate a common preference vector $\bp_{\mathrm{common}} \in [0,1]^{\EH}$ uniformly at random; (2) generate an vector $\bp_d \in [0,1]^{\EH}$ for each resident $d\in \ED$ uniformly at random; (3) calculate each resident's preference by $(1-\alpha) \bp_d + \alpha \bp_{\mathrm{common}}$ where $\alpha\in [0,1]$ is the parameter that controls the correlation level of resident preferences; (4) calculate the value $\rank_d(h)$ as the order of the features in $\bp_d$.
For Set~2 experiments, we generate $b_h\in [0, B)$ uniformly at random satisfying $\sum_{h\in \EH} b_h\in [B, B\times H)$.

We set $D=1,000$, and conduct experiments varying the parameters $D$, $B$, and $\alpha$.
For each combination of parameters, we average the results for $10$ instances.

% \noindent\textbf{Resident preferences:} The preference of each resident consists of the sum of an individual vector and a common vector:
% There is a common preference vector $\bp_{\mathrm{common}} \in [0,1]^C$, where each of the $C$ features is drawn from aa uniform distribution. For each resident $i$, $\bp_i \in [0,1]^C$, which is also drawn in the same manner. The preference vector of resident $i$ is
% \[
% (1-\alpha) \bp_i + \alpha \bp_{\mathrm{common}},
% \]
% where $\alpha \in [0,1]$ determines the influence of the common preference. 
% The value $\rank_i(j)$ is the order of the features in this vector.

\subsection{Details of real data }
\label{sec_real_dgp}

We tested the performance of the proposed method in the dataset of  Japan Residency Matching Program (JRMP) 2007 \cite{kamada:aer:2015}, which matches medical hospital doctors with residency training programs. For ease of discussion, we call each training program ``hospital''. 
Unlike its U.S. counterpart (i.e., the National Resident Matching Program), JRMP does not have ``match variation'' (e.g., consideration of married couples) and adopts the resident-proposing DA algorithm. There are approximately 10,000 doctors in JRMP 2007. Due to computational limitations, we only used a subset of the data; each region has a regional cap, and we focus on the Tokyo region that admits $1,287$ doctors. The original data contains $123$ hospitals, and we clustered them into $C = 50$ (batched) hospitals. 
Each hospital $h$ has a maximum number of admissions $a_h$, which we set as $q_h + b_h = a_h$. We set $q_h$ to be proportional to $a_h$ such that $\sum_h a_h = 1,287$. The dataset also contains the number of applications for each hospital. There are total $6,233$ applications ($4.84$ per hospital). Unfortunately, we do not have the information of which doctor applied to which hospital: We randomly allocate applications such that (1) each doctor applies to at most $8$ hospitals\footnote{Which aligns the original data \cite{kamada:aer:2015}.} and (2) the number of applications for each hospital is the same as the original data. 
Since the preference is partial (i.e., each doctor only ranks a very limited number of hospitals), we introduce a dummy hospital that has infinite capacity and its rank for each doctor is immediately after the least preferred hospital that the doctor applies. By the definition of stable matching, doctors who do not match the hospitals they applied match the dummy hospital.

In summary, we have $1,287$ residents (doctors) who match $50+1$ ($1$ for a dummy) hospitals with a doctor-side preference. We set $B \in \{10, 30\}$ for the limit of capacity expansion.
For each $B$, we average the results for $10$ instances.

\begin{table*}[t]
    \centering
    \small
    \begin{tabular}{|rrr||r|r||r|r|r|r|r|r|r|} \hline
        \multicolumn{1}{|c}{$H$} & \multicolumn{1}{c}{$B$} & \multicolumn{1}{c||}{$\alpha$} & \multicolumn{2}{c||}{Baseline} & \multicolumn{7}{c|}{\textbf{UCT (proposed)}} \\ \cline{4-12}
        &&& \multicolumn{1}{c|}{LPH} & \multicolumn{1}{c||}{Grdy} & \multicolumn{1}{c|}{Iterative} & \multicolumn{1}{c|}{IPT-R} & \multicolumn{1}{c|}{IPT-P} & \multicolumn{1}{c|}{IPT-E} & \multicolumn{1}{c|}{BT-R} & \multicolumn{1}{c|}{BT-P} & \multicolumn{1}{c|}{BT-E} \\ \hline
        $15$ & $30$ & $0.0$ & $20.3$ & $28.0$ & $18.5$ & $25.9$ & $13.6$ &  $5.0$ & $18.2$ & $13.5$ &  $4.4$ \\ 
        $15$ & $30$ & $0.2$ &  $4.1$ &  $3.5$ &  $2.6$ &  $3.4$ &  $1.2$ &  $1.5$ &  $2.0$ &  $1.1$ &  $1.0$ \\ 
        $15$ & $30$ & $0.4$ &  $1.5$ &  $1.0$ &  $0.5$ & $0.02$ & $-0.4$ & $-0.4$ & $0.00$ & $-0.8$ & $-0.9$ \\ \hline
    \end{tabular}
    \caption{Average percentage gaps between the solution found by the Agg-Lin and the solution found by each method for Set~2 experiments. }
    \label{tb:average_gap_synthtwo}
\end{table*}
% \end{landscape}

\begin{table*}[t]
    \centering
    \small
    \begin{tabular}{|rrr||r|r|r||r|r|r|r|r|r|r|} \hline
        \multicolumn{1}{|c}{$H$} & \multicolumn{1}{c}{$B$} & \multicolumn{1}{c||}{$\alpha$} & \multicolumn{3}{c||}{Baseline} & \multicolumn{7}{c|}{\textbf{UCT (proposed)}} \\ \cline{4-13}
        &&& \multicolumn{1}{c|}{Agg-Lin} & \multicolumn{1}{c|}{LPH} & \multicolumn{1}{c||}{Grdy} & \multicolumn{1}{c|}{Iterative} & \multicolumn{1}{c|}{IPT-R} & \multicolumn{1}{c|}{IPT-P} & \multicolumn{1}{c|}{IPT-E} & \multicolumn{1}{c|}{BT-R} & \multicolumn{1}{c|}{BT-P} & \multicolumn{1}{c|}{BT-E} \\ \hline
        $15$ & $30$ & $0.0$ &  $247.95$ & $0.02$ & $2.56$ & $143.65$ & $137.46$ & $136.33$ & $127.89$ & $111.65$ & $114.10$ & $109.93$ \\ 
        $15$ & $30$ & $0.2$ & $2752.07$ & $0.03$ & $6.02$ & $426.97$ & $413.66$ & $412.17$ & $416.50$ & $356.22$ & $389.07$ & $389.58$ \\ 
        $15$ & $30$ & $0.4$ & $3600.08$ & $0.04$ & $9.38$ & $657.50$ & $624.13$ & $625.07$ & $632.34$ & $592.58$ & $618.25$ & $619.15$ \\ \hline
    \end{tabular}
    \caption{Average run times for Set~2 experiments (seconds).}
    \label{tb:run_time_synthtwo}
\end{table*}

\subsection{Simulation results for Set~2}
\label{results_settwo}

Tables \ref{tb:average_gap_synthtwo} and \ref{tb:run_time_synthtwo} show the results for Set~2. The high-level conclusion of these results is not very different from the results for Set~1.

\end{document}